%% file: Submission.tex
\documentclass[manuscript, screen]{acmart}
\usepackage{booktabs} % For formal tables

\usepackage[ruled]{algorithm2e} % For algorithms

\acmJournal{CIE}
\setcopyright{acmcopyright}
\acmDOI{0000001.0000001}

\begin{document}
% Title portion
\title{Robust Online Speed Scaling With Deadline Uncertainty} 
 
\author{Goonwanth Reddy}
%\authornote{This is the corresponding author}
%\orcid{1234-5678-9012-3456}
\affiliation{%
  \institution{Indian Institute of Technology, Madras}
  \streetaddress{Department of Electrical Engineering, Indian Institute of Technology, Madras}
  \city{Chennai, India}}
\email{ngoonwanth@gmail.com}
\author{Rahul Vaze}
%\authornote{This is the corresponding author}
%\orcid{1234-5678-9012-3456}
\affiliation{%
  \institution{Tata Institute of Fundamental Research, Mumbai}
  \streetaddress{School of Technology and Computer Science, Tata Institute of Fundamental Research}
  \city{Mumbai, India,}}
\email{vaze@tcs.tifr.res.in}
\renewcommand\shortauthors{Reddy, G. et al}

\begin{abstract}
A speed scaling problem is considered, where time is divided into slots, and jobs with payoff $v$ arrive at the beginning of the slot with associated deadlines $d$. Each job takes one slot to be processed, and multiple jobs can be processed by the server in each slot with energy cost $g(k)$ for processing $k$ jobs in one slot. The payoff is accrued by the algorithm only if the job is processed by its deadline. We consider 
a robust version of this speed scaling problem, where a job on its arrival reveals its payoff $v$, however, the deadline is hidden to the online algorithm, which could potentially be chosen adversarially and known to the optimal offline algorithm. The objective is to derive a robust (to deadlines) and optimal online algorithm that achieves the best competitive ratio. We propose an algorithm (called min-LCR) and show that it is an optimal online algorithm for any convex energy cost function $g(.)$. We do so without actually evaluating the optimal competitive ratio,  and give a general proof that works for any convex $g$, which is rather novel. 
For the popular choice of energy cost function $g(k) = k^\alpha, \alpha \ge 2$, we give concrete bounds on the competitive ratio of the algorithm, which ranges between $2.618$ and $3$ depending on the value of $\alpha$. The best known online algorithm for the same problem, but where deadlines are revealed to the online algorithm has competitive ratio of $2$ and a lower bound of $\sqrt{2}$. Thus, importantly, lack of deadline knowledge does not make the problem degenerate, and the effect of deadline information on the optimal competitive ratio is limited.
\end{abstract}

%
% End generated code
%

\keywords{Online Algorithms, Speed Scaling, Greedy Algorithms, Scheduling}

\maketitle

\input{input.tex}

\input{Draft1}

\end{document}

%% file: input.tex
%\usepackage{fancyhdr}
%\pagestyle{fancy}

%\usepackage{amsfonts}
%\usepackage{times}
%%\usepackage[pdftex]{graphicx}
%%\DeclareGraphicsExtensions{.jpg}
%%\usepackage[dvips]{graphicx}
%%\DeclareGraphicsExtensions{.eps}
%\usepackage{latexsym}
%\usepackage{amssymb}
%\usepackage{amsmath}
%\usepackage{cite}
%\usepackage{verbatim}
%\newtheorem{theorem}{Theorem}
%
%\newtheorem{acknowledgement}[theorem]{Acknowledgement}
%\newtheorem{algorithm}[theorem]{Algorithm}
%\newtheorem{assumption}[theorem]{Assumption}
%\newtheorem{axiom}[theorem]{Axiom}
%\newtheorem{case}[theorem]{Case}
%\newtheorem{claim}[theorem]{Claim}
%\newtheorem{conclusion}[theorem]{Conclusion}
%\newtheorem{condition}[theorem]{Condition}
%\newtheorem{conjecture}[theorem]{Conjecture}
%\newtheorem{corollary}[theorem]{Corollary}
%\newtheorem{criterion}[theorem]{Criterion}
%\newtheorem{definition}[theorem]{Definition}
%\newtheorem{example}[theorem]{Example}
%\newtheorem{exercise}[theorem]{Exercise}
%\newtheorem{fact}[theorem]{Fact}
%\newtheorem{lemma}[theorem]{Lemma}
%\newtheorem{notation}[theorem]{Notation}
%\newtheorem{problem}[theorem]{Problem}
%\newtheorem{proposition}[theorem]{Proposition}
%\newtheorem{solution}[theorem]{Solution}
%\newtheorem{summary}[theorem]{Summary}
\newtheorem{remark}[theorem]{Remark}

\newcommand{\figref}[1]{{Fig.}~\ref{#1}}
\newcommand{\tabref}[1]{{Table}~\ref{#1}}
\newcommand{\bookemph}[1]{ {\em #1}}
\newcommand{\Ns}{N_s}
\newcommand{\Ut}{U_t}
\newcommand{\fig}[1]{Fig.\ \ref{#1}}
\def\onehalf{\frac{1}{2}}
\def\etal{et.\/ al.\/}
\newcommand{\bydef}{\triangleq}
\newcommand{\tr}{{\it{tr}}}
\def\SNR{{\textsf{SNR}}}
\def\Pe{{P_e}}
\def\SINR{{\mathsf{SINR}}}
\def\SIR{{\mathsf{SIR}}}
\def\MI{{\mathsf{MI}}}
% blackboard lowercase
\def\bydef{:=}
\def\bba{{\mathbb{a}}}
\def\bbb{{\mathbb{b}}}
\def\bbc{{\mathbb{c}}}
\def\bbd{{\mathbb{d}}}
\def\bbee{{\mathbb{e}}}
\def\bbff{{\mathbb{f}}}
\def\bbg{{\mathbb{g}}}
\def\bbh{{\mathbb{h}}}
\def\bbi{{\mathbb{i}}}
\def\bbj{{\mathbb{j}}}
\def\bbk{{\mathbb{k}}}
\def\bbl{{\mathbb{l}}}
\def\bbm{{\mathbb{m}}}
\def\bbn{{\mathbb{n}}}
\def\bbo{{\mathbb{o}}}
\def\bbp{{\mathbb{p}}}
\def\bbq{{\mathbb{q}}}
\def\bbr{{\mathbb{r}}}
\def\bbs{{\mathbb{s}}}
\def\bbt{{\mathbb{t}}}
\def\bbu{{\mathbb{u}}}
\def\bbv{{\mathbb{v}}}
\def\bbw{{\mathbb{w}}}
\def\bbx{{\mathbb{x}}}
\def\bby{{\mathbb{y}}}
\def\bbz{{\mathbb{z}}}
\def\bb0{{\mathbb{0}}}

% Bold lowercase
\def\bydef{:=}
\def\ba{{\mathbf{a}}}
\def\bb{{\mathbf{b}}}
\def\bc{{\mathbf{c}}}
\def\bd{{\mathbf{d}}}
\def\bee{{\mathbf{e}}}
\def\bff{{\mathbf{f}}}
\def\bg{{\mathbf{g}}}
\def\bh{{\mathbf{h}}}
\def\bi{{\mathbf{i}}}
\def\bj{{\mathbf{j}}}
\def\bk{{\mathbf{k}}}
\def\bl{{\mathbf{l}}}
\def\bm{{\mathbf{m}}}
\def\bn{{\mathbf{n}}}
\def\bo{{\mathbf{o}}}
\def\bp{{\mathbf{p}}}
\def\bq{{\mathbf{q}}}
\def\br{{\mathbf{r}}}
\def\bs{{\mathbf{s}}}
\def\bt{{\mathbf{t}}}
\def\bu{{\mathbf{u}}}
\def\bv{{\mathbf{v}}}
\def\bw{{\mathbf{w}}}
\def\bx{{\mathbf{x}}}
\def\by{{\mathbf{y}}}
\def\bz{{\mathbf{z}}}
\def\b0{{\mathbf{0}}}
\def\opt{\mathsf{OPT}}
\def\alg{\mathsf{ALG}}
\def\offalg{\mathsf{OFF-ALG}}
\def\on{\mathsf{ON}}
\def\off{\mathsf{OFF}}
\def\lcr{\mathsf{LCR}}
\def\greedy{\mathsf{Greedy}}
% Bold capital letters
\def\bA{{\mathbf{A}}}
\def\bB{{\mathbf{B}}}
\def\bC{{\mathbf{C}}}
\def\bD{{\mathbf{D}}}
\def\bE{{\mathbf{E}}}
\def\bF{{\mathbf{F}}}
\def\bG{{\mathbf{G}}}
\def\bH{{\mathbf{H}}}
\def\bI{{\mathbf{I}}}
\def\bJ{{\mathbf{J}}}
\def\bK{{\mathbf{K}}}
\def\bL{{\mathbf{L}}}
\def\bM{{\mathbf{M}}}
\def\bN{{\mathbf{N}}}
\def\bO{{\mathbf{O}}}
\def\bP{{\mathbf{P}}}
\def\bQ{{\mathbf{Q}}}
\def\bR{{\mathbf{R}}}
\def\bS{{\mathbf{S}}}
\def\bT{{\mathbf{T}}}
\def\bU{{\mathbf{U}}}
\def\bV{{\mathbf{V}}}
\def\bW{{\mathbf{W}}}
\def\bX{{\mathbf{X}}}
\def\bY{{\mathbf{Y}}}
\def\bZ{{\mathbf{Z}}}
\def\b1{{\mathbf{1}}}

% Blackboard capital letters
\def\bbA{{\mathbb{A}}}
\def\bbB{{\mathbb{B}}}
\def\bbC{{\mathbb{C}}}
\def\bbD{{\mathbb{D}}}
\def\bbE{{\mathbb{E}}}
\def\bbF{{\mathbb{F}}}
\def\bbG{{\mathbb{G}}}
\def\bbH{{\mathbb{H}}}
\def\bbI{{\mathbb{I}}}
\def\bbJ{{\mathbb{J}}}
\def\bbK{{\mathbb{K}}}
\def\bbL{{\mathbb{L}}}
\def\bbM{{\mathbb{M}}}
\def\bbN{{\mathbb{N}}}
\def\bbO{{\mathbb{O}}}
\def\bbP{{\mathbb{P}}}
\def\bbQ{{\mathbb{Q}}}
\def\bbR{{\mathbb{R}}}
\def\bbS{{\mathbb{S}}}
\def\bbT{{\mathbb{T}}}
\def\bbU{{\mathbb{U}}}
\def\bbV{{\mathbb{V}}}
\def\bbW{{\mathbb{W}}}
\def\bbX{{\mathbb{X}}}
\def\bbY{{\mathbb{Y}}}
\def\bbZ{{\mathbb{Z}}}

% Caligraphic capital letters
\def\cA{\mathcal{A}}
\def\cB{\mathcal{B}}
\def\cC{\mathcal{C}}
\def\cD{\mathcal{D}}
\def\cE{\mathcal{E}}
\def\cF{\mathcal{F}}
\def\cG{\mathcal{G}}
\def\cH{\mathcal{H}}
\def\cI{\mathcal{I}}
\def\cJ{\mathcal{J}}
\def\cK{\mathcal{K}}
\def\cL{\mathcal{L}}
\def\cM{\mathcal{M}}
\def\cN{\mathcal{N}}
\def\cO{\mathcal{O}}
\def\cP{\mathcal{P}}
\def\cQ{\mathcal{Q}}
\def\cR{\mathcal{R}}
\def\cS{\mathcal{S}}
\def\cT{\mathcal{T}}
\def\cU{\mathcal{U}}
\def\cV{\mathcal{V}}
\def\cW{\mathcal{W}}
\def\cX{\mathcal{X}}
\def\cY{\mathcal{Y}}
\def\cZ{\mathcal{Z}}

% Sans serif capital letters
\def\sfA{\mathsf{A}}
\def\sfB{\mathsf{B}}
\def\sfC{\mathsf{C}}
\def\sfD{\mathsf{D}}
\def\sfE{\mathsf{E}}
\def\sfF{\mathsf{F}}
\def\sfG{\mathsf{G}}
\def\sfH{\mathsf{H}}
\def\sfI{\mathsf{I}}
\def\sfJ{\mathsf{J}}
\def\sfK{\mathsf{K}}
\def\sfL{\mathsf{L}}
\def\sfM{\mathsf{M}}
\def\sfN{\mathsf{N}}
\def\sfO{\mathsf{O}}
\def\sfP{\mathsf{P}}
\def\sfQ{\mathsf{Q}}
\def\sfR{\mathsf{R}}
\def\sfS{\mathsf{S}}
\def\sfT{\mathsf{T}}
\def\sfU{\mathsf{U}}
\def\sfV{\mathsf{V}}
\def\sfW{\mathsf{W}}
\def\sfX{\mathsf{X}}
\def\sfY{\mathsf{Y}}
\def\sfZ{\mathsf{Z}}

% sans serif lowercase
\def\bydef{:=}
\def\sfa{{\mathsf{a}}}
\def\sfb{{\mathsf{b}}}
\def\sfc{{\mathsf{c}}}
\def\sfd{{\mathsf{d}}}
\def\sfee{{\mathsf{e}}}
\def\sfff{{\mathsf{f}}}
\def\sfg{{\mathsf{g}}}
\def\sfh{{\mathsf{h}}}
\def\sfi{{\mathsf{i}}}
\def\sfj{{\mathsf{j}}}
\def\sfk{{\mathsf{k}}}
\def\sfl{{\mathsf{l}}}
\def\sfm{{\mathsf{m}}}
\def\sfn{{\mathsf{n}}}
\def\sfo{{\mathsf{o}}}
\def\sfp{{\mathsf{p}}}
\def\sfq{{\mathsf{q}}}
\def\sfr{{\mathsf{r}}}
\def\sfs{{\mathsf{s}}}
\def\sft{{\mathsf{t}}}
\def\sfu{{\mathsf{u}}}
\def\sfv{{\mathsf{v}}}
\def\sfw{{\mathsf{w}}}
\def\sfx{{\mathsf{x}}}
\def\sfy{{\mathsf{y}}}
\def\sfz{{\mathsf{z}}}
\def\sf0{{\mathsf{0}}}

\def\Nt{{N_t}}
\def\Nr{{N_r}}
\def\Ne{{N_e}}
\def\Ns{{N_s}}
\def\Es{{E_s}}
\def\No{{N_o}}
\def\sinc{\mathrm{sinc}}
\def\dmin{d^2_{\mathrm{min}}}
\def\vec{\mathrm{vec}~}
\def\kron{\otimes}
\def\Pe{{P_e}}
\newcommand{\expeq}{\stackrel{.}{=}}
\newcommand{\expg}{\stackrel{.}{\ge}}
\newcommand{\expl}{\stackrel{.}{\le}}
\def\SIR{{\mathsf{SIR}}}

% Added by Takao
\def\nn{\nonumber}

%% file: Draft1.tex
\begin{abstract}

\end{abstract}

%%%%%%
%%%%%%
\newpage
\section{Introduction}

Energy efficient transmission of packets in communication systems or processing of jobs in microprocessors (and similar applications) is a fundamental resource allocation problem. A job/packet can be processed/transmitted by a server `fast' but only at the cost of higher energy consumption. Typically, the energy cost is a convex function of the server speed, (e.g., a popular choice is $x^\alpha$ for $\alpha \ge 2$) and the general objective is two fold: maximize the profit from processing jobs while incurring minimum energy cost. The profit can either be the payoff accrued on processing a job or some function of the inverse of the processing time. 
This problem is known as {\it speed scaling} problem in literature, where the speed of the server is the tunable parameter which determines the profit and the energy consumption.

Speed scaling problem has been considered with both the infinite speed model as well as bounded speed model, where in the former case, the server is allowed to scale its speed anywhere in $[0, \infty)$, while in the latter case, it is bounded by a fixed constant. Infinite speed model allows processing of all jobs, and the main concern is to minimize energy consumption, while in the bounded speed model, both maximizing the profit and minimizing the energy consumption is a challenge. 

Another broad classification considered with the speed scaling problem is with respect to deadlines. In the case when jobs do not have deadlines, a typical objective is to minimize a linear combination of the sum of the flow time (completion time minus the arrival time) and the energy consumption for each job. When jobs also specify a deadline, under the bounded speed model, the objective is to process jobs that maximize the profit while minimizing the energy cost.  

An alternate speed scaling model (which we consider in this paper) is a discretized one, where time is divided into slots, and jobs with payoff $v$ arrive at the beginning of the slot with associated deadlines $d$. Each job takes one slot to be processed, and multiple jobs can be processed by the server in each slot with energy cost $g(k)$ for processing $k$ jobs in one slot. The payoff is accrued by the algorithm only if the job is processed by its deadline, and the objective is to maximize the sum of the profit (payoff minus energy cost) of the processed jobs.  

One limiting aspect of almost all literature on speed scaling with job deadlines is the need for the exact knowledge of deadlines. The derived results critically depend on exact deadline information, and are not robust to even a small uncertainty. In modern applications, the job deadlines could be time varying, could potentially depend on other jobs or their completion times, or may not be precisely known on the job arrival.
Towards addressing this critical aspect as well as to generalize the model, we consider a robust version of the discretized speed scaling problem, where a job on its arrival reveals its payoff $v$, however, the deadline is hidden to the online algorithm, which could potentially be chosen adversarially and known to offline optimal algorithm. 
This approach will lead to the derivation of robust online algorithms for speed scaling and provides a means to quantify the fundamental effect of deadline knowledge on speed scaling. We note that the robust approach is useful only if the problem itself does not become degenerate, in the sense that no online algorithm can achieve a reasonable competitive ratio. 

For the considered robust model, we propose a simple online algorithm, called the minimum-local competitive ratio (min-LCR) algorithm, that does not need any deadline information. Let $c_k = g(k)-g(k-1)$ be the effective energy cost for processing the $k^{th}$ job. Algorithm min-LCR indexes all the available jobs at each slot in the non-increasing order of their payoffs $v$, 
and computes the profit $p_\ell$ it will make if it processes $1\le \ell \le m$ jobs, where $m$ is such that $v_m-c_m >0$ and $v_{m+1}-c_{m+1} < 0$. It also presumes a worst case scenario for the deadlines (since it does not know the exact deadlines) where the $\ell$-chosen jobs for processing have deadlines infinity while the left-over jobs (other than the $\ell$ chosen jobs) have deadlines that expire in the current slot. Let $o_k$ be the profit that can be made by an offline algorithm under the knowledge of the worst choice of deadlines for the current set of available jobs, assuming no further jobs arrive. We define $\lcr_\ell = \frac{o_\ell}{p_\ell}$ for $1\le \ell \le m$, which has the interpretation of local competitive ratio (ratio of the optimal offline and the simple online algorithm under the worst case deadlines given no further jobs arrive), and the algorithm chooses to process $\ell$ jobs for which $\lcr_\ell$ is minimum. Note that always choosing $\ell=m$ will correspond to a natural greedy algorithm for this problem.

%%%%%%
%%%%%%

\subsection{Contributions}
We make the following contributions in this paper.
\begin{itemize}
\item We show that the proposed Algorithm min-LCR is an optimal online algorithm for any convex energy cost function $g(.)$. We do so without actually finding the optimal competitive ratio, which is fundamentally different than the typical proof strategy used in the analysis of online algorithms. To put our result's importance in perspective, to the best of our knowledge, no other online algorithm in the vast speed scaling literature is known to achieve the optimal competitive ratio except in rare cases where energy cost function is $g(k) = k^\alpha$ and when deadlines are exactly known. 

\item To ensure that the considered robust speed scaling problem is not degenerate in the sense that the adversarial deadlines choices that are unknown to the online algorithm make the competitive ratio of any online algorithm arbitrarily large, we provide concrete analysis for the most popular energy cost function of $g(k) = k^\alpha$. The min-LCR algorithm involves finding the optimal number of jobs to process among the available ones that minimizes the $\lcr$. We also consider a simplified version of the min-LCR algorithm, where exactly $k=\left \lfloor{\beta m}\right \rfloor$ or $k=\left \lceil{\beta m}\right \rceil$ jobs are processed in each slot depending on whichever one has lower LCR, where $\beta$ is the solution of the equation $x^\alpha + x^{\alpha-1} -1=0$. We also consider the natural greedy special case of min-LCR that always processes $m$ jobs, where $m$ has been defined in the min-LCR algorithm description as above.
\item We show that the optimal competitive ratio for energy cost function of $g(k) = k^\alpha$ with $\alpha=2$ is $\phi+1$ ($\phi = 1/\delta, \delta = \frac{\sqrt{5}-1}{2}$), and is achieved by the simplified min-LCR Algorithm. In comparison, when deadlines are known to the online algorithm, for $\alpha\ge 2$ the best known online algorithm has competitive ratio of $2$ and the best known lower bound is $\sqrt{2}$. 
\item For $\alpha\ge  2.5$, the competitive ratio of the simplified min-LCR Algorithm is at most $\phi+1$, and for any $\alpha \ge 2$, the competitive ratio of the greedy algorithm (min-LCR Algorithm with $\ell=m$ always) is at most $3$. We also derive a lower bound on the competitive ratio of $\sqrt{2}+1$ for all $\alpha \ge 2$. 
\item Thus, we show that for the energy cost function of $g(k) = k^\alpha$, lack of deadline knowledge reduces the optimal competitive ratio by a factor of at most $3/\sqrt{2}$. Thus, the loss in performance because of deadline uncertainty is limited, and precludes the possibility that the considered robust model is inherently weak, and the power of any online algorithm is seriously limited.

\end{itemize}

\subsection{Related Work}
Starting with \cite{yao1995scheduling}, there has been a long line of work on optimal speed scaling for servers with unbounded speed. For energy cost $g(k) = k^{\alpha}, \alpha >1$, a $2^{\alpha-1}\alpha^\alpha$-competitive algorithm was derived in \cite{yao1995scheduling}, whose competitive ratio was subsequently improved upon in \cite{bansal2007speed} to $2 \left(\frac{\alpha}{\alpha-1}\right)^{\alpha}\exp(1)^\alpha$, and eventually in \cite{bansal2009improved} to  $4^\alpha/ (2\sqrt{\exp(1) \alpha})$.

For the bounded speed case, where all jobs cannot be processed, \cite{chan2009optimizing} first derived an online algorithm that is $14$-competitive algorithm for throughput (number of processed jobs) and $\alpha^\alpha + \alpha^24^\alpha$-competitive for energy. Subsequently, the throughput competitiveness was improved to $4$ in \cite{bansal2008scheduling}, which is also the best possible \cite{baruah1991line}. 

In addition to speed scaling, an additional feature of {\it sleeping} was introduced in \cite{irani2007algorithms}, where all jobs can be processed with $(2^{2\alpha-2} \alpha^\alpha + 2^{\alpha-1}+2)$-competitiveness in terms of energy, which was improved upon in 
\cite{han2010deadline} to get $\alpha^\alpha +2$-competitiveness in terms of energy under the infinite speed model when all jobs can be processed, and $4$-competitive algorithm for throughput (number of processed jobs) and $(\alpha^\alpha + \alpha^24^\alpha+2)$-competitive for energy in the bounded server model. 

The no-deadline model where the objective function is to minimize the flow time plus the energy has been considered widely \cite{lam2008speed, wierman2009power, bansal2009speed}, with the most general result obtained in \cite{bansal2009speed} that gives a constant competitive algorithm for all energy cost functions. 
The speed scaling problem in the no-deadline model, where some information about jobs (either value/weight/density) is hidden is called the non-clairvoyant setting and has been addressed in literature starting with \cite{chan2011nonclairvoyant} and followed up in \cite{azar2015speed}. Speed scaling with multiple processors has also been considered in \cite{albers2014speed} and with non-clairvoyant setting in \cite{gupta2011nonclairvoyantly}, while modern applications for speed scaling are being addressed in \cite{Barcelo2016}, where energy is derived from solar cells for renewable energy harvesting. 

The discrete model studied in this paper was first considered in \cite{cote2010energy}, where jobs deadlines are known to the online algorithm, which proposed an online algorithm that is $2$-competitive, using the ideas from online request matching \cite{riedel2001online}. An associated lower bound (easy to construct) on the competitive ratio for this problem is $\sqrt{2}$.

As far as we know the speed scaling problem when jobs have hard deadlines that are not exactly known to the online algorithm has not been considered in literature. In load balancing literature, unknown job deadline case is referred to as scheduling for temporary tasks, where the duration for which a job lasts is unknown \cite{azar1993online}. In load-balancing, however, each job has to be scheduled as soon as it arrives, and the only decision variables are : which server to be assigned for each job and the dynamic server speed.

%%%%%%
%%%%%%

\section{Problem Definition}

%%%%%%
%%%%%%

We consider a discrete time system, where time is divided in discrete slots.
A sequence $\boldsymbol{\sigma} = J_1, \dots, J_n$ of jobs arrives causally, where job $J_i$ 
arrives at slot/time $a_i$ and must be finished by a deadline of $d_i$ slots starting from $a_i$ or is dropped. We assume $a_i \le a_{i+1}$ for $1\le i < n$. Each job takes one slot to be processed, and if processed before its deadline, job $J_i$ accumulates a payoff/value $v_i$. A job is {\it available} at slot $t$ if its absolute deadline is after slot $t-1$, and is {\it expired} otherwise.
The server can work at variable speed, and can process $k\ge 0$ jobs in any slot by incurring a cost of $g(k)$, where $g(.)$ is a convex function, e.g., $g(k) = k^{\alpha}, \alpha > 1$. 

In a significant departure from prior work on speed scaling, we consider the robust setting, where the online algorithm does not know the deadline for any job, which could potentially be chosen by an adversary. 
This allows us to model the deadline uncertainty, derive a robust online algorithm, and provide a means to quantify the effect of deadline knowledge on speed scaling. Thus, the information that any online algorithm has at slot $t$ is the set of jobs that have not expired by then, and their respective values. We, however, let the offline optimal algorithm to know the exact deadline and the respective payoff non-causally, to consider the worst case model. 

We consider only the deterministic algorithm setting, where on a sequence of jobs $\boldsymbol{\sigma}$, let the set of jobs processed at slot $j$ by an algorithm $\alg$ be $P_j$. Then the overall profit for $\alg$ is $$C_{\alg}(\boldsymbol{\sigma}) = \sum_{j=1}^{\mathsf{last}} (v_{P_j} - g(|P_j|)),$$
where $v_{P_j} = \sum_{i\in P_j} v_j$, and $\mathsf{last}$ is the last slot at which Algorithm $\alg$ processes any job. The objective is to minimize the competitive ratio $\sfr_{\alg} = \max_{\boldsymbol{\sigma}}\frac{C_{\off}(\boldsymbol{\sigma})}{C_{\alg}(\boldsymbol{\sigma})}$,
where $\off$ is the offline optimal algorithm that is allowed to know the sequence $\boldsymbol{\sigma}$ including the job deadlines non-causally. Let the optimal competitive ratio be 
\begin{equation}\label{eq:optcompratio}
\sfr^\star = \min_{\alg} \max_{\boldsymbol{\sigma}} \frac{C_{\off}(\boldsymbol{\sigma})}{C_{\alg}(\boldsymbol{\sigma})}, \end{equation}
and the optimal online algorithm achieving $\sfr^\star$ be $\opt$. To reiterate, $\min_{\alg}$ is over all deterministic online algorithms that do not have deadline information of jobs, and make decisions causally at each slot depending on the values of the available jobs only.
\begin{definition} The effective/incremental cost of processing the $k^{th}$ job in any slot is $c_k = g(k)-g(k-1)$. Since $g(k)$ is convex, $c_k > 0 \ \forall \ k$.
\end{definition}

\begin{definition} For two inputs $\boldsymbol{\sigma}_1$ and 
$\boldsymbol{\sigma}_2$, $\boldsymbol{\sigma}_1 \cup \boldsymbol{\sigma}_2$ corresponds to input where for each slot $t$, the set of jobs is the union of job arriving at slot $t$ in $\boldsymbol{\sigma}_1$ and 
$\boldsymbol{\sigma}_2$.
\end{definition}

\begin{lemma}\label{lem:subadd} For any two inputs $\boldsymbol{\sigma}_1$ and 
$\boldsymbol{\sigma}_2$, 
$C_{\off}(\boldsymbol{\sigma}_1 \cup \boldsymbol{\sigma}_2) \leq C_{\off}(\boldsymbol{\sigma}_1) + C_{\off}(\boldsymbol{\sigma}_2).$ 
\end{lemma}
\begin{proof}
Let the profit (payoff minus the energy cost) accrued in $C_{\off}(\boldsymbol{\sigma}_1 \cup \boldsymbol{\sigma}_2)$ corresponding to the processing of jobs belonging to $\boldsymbol{\sigma}_i$ be $p_i$. Then $p_i \le C_{\off}(\boldsymbol{\sigma}_i)$ and the result follows since $C_{\off}(\boldsymbol{\sigma}_1 \cup \boldsymbol{\sigma}_2) = p_1+p_2$.\end{proof}

%%%%%%
%%%%%%
\section{min-LCR Algorithm}

\begin{algorithm}[ht]
%\label{algorithm:LCR}
\caption{min-LCR Algorithm}
\SetKwInOut{Input}{Input}
\Input{Sequence $\boldsymbol{\sigma} = J_1, \dots, J_n$}
At slot $\tau$, consider the union of all the non-processed jobs by the algorithm so far that are available and the newly arriving jobs in slot $\tau$, and call it $E(\tau)$ \\
Arrange the jobs in $E(\tau)$ in non-increasing order of their payoffs/value $v$ \\
Let $E(\tau)^i = \{\text{first} \ \ i \ \text{jobs in} \ E(\tau)\}$ and $E^{i-}(\tau) = E(\tau) \backslash E(\tau)^i$\\
Let $v^{(i)}$ be the value of job of $E(\tau)$ with the $i^{th}$ highest value and 
$V^{(i)}$ be the sum of the values of the first $i$ jobs with highest values\\
Let $m =  \max_{\{v^{(j)} - [g(j)-g(j-1)] > 0\}} j $ \\
\For{$i=1:m$}{
	Let 
$M_{i,\tau} = V^{(i)} - i g(1),  P_{i,\tau} = V^{(i)} - g(i), C_{\greedy}(i,\tau) = \max_{j \in E^{i-}(\tau) } \{ {V^{(j)} - g(j)} \} $\\	
$\lcr_{i}(\tau) = \dfrac{M_{i,\tau} + C_{\greedy}(i,\tau)}{P_{i,\tau}}$\\
	
		}	
		Let $i^\star(\tau) = \arg \min_{i=1, \dots, m} \lcr_{i}(\tau)$\\
		Process first $i^\star(\tau)$ jobs of $E(\tau)$ at slot $\tau$, call the processed set of jobs $L^p_{\boldsymbol{\sigma}}(\tau)$.
\end{algorithm}
\begin{remark}\label{rem:greedysp} Since $c_k > 0$ and increasing in $k$, $m$ (the largest number of jobs that can be processed, where each job has a positive profit) is well defined in min-LCR algorithm.
\end{remark}
\begin{remark} A natural $\greedy$ algorithm is a special case of the min-LCR Algorithm when $i^\star(\tau) = m$ at all slots $\tau$.
\end{remark}
The basic idea behind the min-LCR algorithm has been described in the introduction. To be specific,  the online algorithm's profit is $P_{i,\tau}$ if it processes $i$ jobs in slot $\tau$ without the knowledge of the deadlines of the available jobs. The online algorithm presumes that possibly no further jobs are going to arrive, and the $i$-chosen jobs for processing by the algorithm have deadlines as infinity, while the
left-over jobs (other than the $i$ chosen jobs) have deadlines that expire in the current slot.
The profit of the $\off$ under this presumption is $o_{i,\tau} = M_{i,\tau} + C_{\greedy}(i,\tau)$, where $M_{i,\tau}$ is the profit $\off$ can accrue by processing $i$ highest valued jobs one in each slot starting from slot $\tau+1$ if their deadlines are infinity, while $C_{\greedy}(i,\tau)$ is the largest profit possible by processing the set of jobs other than the $i$ highest valued jobs in slot $\tau$ itself. The algorithm chooses $i$ that minimizes the ratio of $\frac{o_{i,\tau}}{P_{i,\tau}}$ which essentially is the local competitive ratio at slot $\tau$.

\begin{theorem}\label{thm:lcr} Algorithm min-LCR is an optimal online algorithm that achieves the optimal competitive ratio \eqref{eq:optcompratio}.
\end{theorem}
To prove Theorem \ref{thm:lcr}, we show that the competitive ratio of Algorithm min-LCR on input $\boldsymbol{\sigma}$ is at most $\max_{\tau} \lcr_{i^\star(\tau)}(\tau)$ in Lemma \ref{lem:compratiolcrub},  and $\max_{\tau} \lcr_{i^\star(\tau)}(\tau) \le \sfr^{\star}$ in Lemma \ref{lem:lcrcomparison}.

\begin{lemma}\label{lem:compratiolcrub} The competitive ratio of Algorithm min-LCR on input $\boldsymbol{\sigma}$ is at most $\max_{\tau} \lcr_{i^\star(\tau)}(\tau)$.
\end{lemma}
\begin{proof}
Let $L^p_{\boldsymbol{\sigma}}(\tau)$ be the jobs processed by Algorithm min-LCR on input $\boldsymbol{\sigma}$ at slot $\tau$, and  $\bigcup\limits_{\tau=1}^{\mathsf{last}} L^p(\tau) = L^p$, where $\mathsf{last}$ is the last slot at which Algorithm min-LCR processes any job. 
%Note that min-LCR processes its last jobs at the same slot or later as the $\off$.
%We will drop the subscript whenever possible to do so without confusion. By definition,
\begin{align}\nn
    C_{\off}(\boldsymbol{\sigma}) &= C_{\off}(L^p \cup  \boldsymbol{\sigma}\backslash L^p),\\\nn
    &\stackrel{(a)}\leq C_{\off}(L^p) + C_{\off}(\boldsymbol{\sigma} \backslash L^p),\\\nn
    &= C_{\off}\bigg(\bigcup\limits_{\tau=1}^{\mathsf{last}} L^p(\tau)  \bigg) + C_{\off}(\boldsymbol{\sigma}\backslash L^p),\\ \label{eq:int1}
    &\stackrel{(b)}\leq \sum_{\tau=1}^{\mathsf{last}} C_{\off}(L^p(\tau)) \;\;+\;\; C_{\off}(\boldsymbol{\sigma}\backslash L^p) = \sum_{\tau=1}^{\mathsf{last}} C_{\off}(L^p(\tau)) \;\;+\;\; \sum_{\tau=1}^{\mathsf{last}} C_{\off}(\boldsymbol{\sigma}\backslash L^p)_{\tau},
\end{align}
where $(a)$ and $(b)$ follow from sub-additivity Lemma (Lemma \ref{lem:subadd}), and where $C_{\off}(\boldsymbol{\sigma}\backslash L^p)_{\tau}$ is the profit obtained by the $\off$ algorithm in slot $\tau$ when the input is only $\boldsymbol{\sigma}\backslash L^p$. 

The largest profit can be made from jobs in set $L^p(\tau)$ if each of them are processed alone in distinct slots with 
energy cost of $g(1)$, hence
\begin{equation}\label{eq:dummy39} 
C_{\off}(L^p(\tau)) \leq \sum_{i \in L^p(\tau)} (v_i -g(1)) \end{equation}
Since $L^p(\tau)$ is the set consisting of first $i^\star(\tau)$ jobs of $E(\tau)$, and $|L^p(\tau)|  = i^\star(\tau)$, we get  $\sum_{i \in L^p(\tau)} v_i = V^{i^\star(\tau)}$ and $\sum_{i \in L^p(\tau)} g(1) = i^\star(\tau)g(1)$. Hence from \eqref{eq:dummy39}, we get 
\begin{equation}\label{eq:dummy3} 
C_{\off}(L^p(\tau)) \leq  V^{i^\star(\tau)} - i^\star(\tau)g(1) = M_{i^\star,\tau},
\end{equation}
where the last inequality follows from the definition of $M_{i^\star,\tau}$ from Algorithm min-LCR.

%Let $\off(\boldsymbol{\sigma}\backslash L^p)_{\tau}$ be the union of jobs that arrive in slot $\tau$ and the set of jobs $\boldsymbol{\sigma}\backslash L^p$ that have not been processed by $\off$ by slot $\tau$ that have also not expired by slot $\tau$. Essentially, 
The set $\boldsymbol{\sigma} \backslash L^p$ consists of all jobs of $\boldsymbol{\sigma}$ that are not processed by the min-LCR algorithm at any slot during its operation. The 
$( \boldsymbol{\sigma} \backslash L^p )_{\tau}$ is the set of elements of $\boldsymbol{\sigma} \backslash L^p$ which arrived at or before slot $\tau$ and whose deadline is after slot $\tau-1$.
In comparison, the set $E(\tau) \backslash L^p {(\tau)}$ is the union of jobs available at slot $\tau$ that are never processed by the min-LCR algorithm and the available jobs at slot $\tau$ that are processed by the min-LCR algorithm in some slot after slot $\tau$.

%is the union of the jobs that are available with the min-LCR algorithm at slot $\tau$ but are not processed by min-LCR algorithm at slot $\tau$, and the jobs that are available at $\tau$ with the min-LCR algorithm that may or may not be  processed by the min-LCR algorithm at some later time.  

%Let
%$\off(\boldsymbol{\sigma}\backslash L^p)_{\tau}$ contains all jobs that were not processed by the min-LCR algorithm that have arrived before or at slot $\tau$ and have deadlines at or after slot $\tau$, and have not been processed by $\off$ on input $\boldsymbol{\sigma}\backslash L^p$ until slot $\tau-1$. Moreover, set $E(\tau) \backslash L^p(\tau)$ is the set of jobs that are not processed by the min-LCR algorithm till the end of slot $\tau$ and have not expired by slot $\tau-1$.
%$E(\tau) \backslash L^p(\tau)$ is the set of jobs that are not processed by min-LCR algorithm until slot $\tau$, that have arrived before slot $\tau$ and have deadlines on or after slot $\tau$. 
Thus, 
\begin{equation}\label{eq:dummy1} (\boldsymbol{\sigma}\backslash L^p)_{\tau} \subseteq E(\tau) \backslash L^p(\tau), 
\end{equation}
%where recall that $E(\tau)$ is the set of all non-processed jobs by min-LCR Algorithm that have not expired by slot $\tau$. 
By definition, 
$C_{\off}(\boldsymbol{\sigma}\backslash L^p)_{\tau} \le \max_{A \in \offalg} C_A(\boldsymbol{\sigma}\backslash L^p)_{\tau}$, where the maximization is over all the offline algorithms. Hence, using \eqref{eq:dummy1}, we get 
\begin{equation}\label{eq:dummy2}
C_{\off}(\boldsymbol{\sigma}\backslash L^p)_{\tau} \le \max_{A \in \offalg} C_A(E(\tau) \backslash L^p(\tau))_{\tau} = C_{\greedy}(i^\star, \tau),
\end{equation} where the last equality is derived as follows.
Let $A^\star$ be the maximizer  of $\max_{A \in \offalg} C_A(E(\tau) \backslash L^p(\tau))_{\tau}$ and say it processes some $k$ jobs with values $v_1, \dots, v_k$ belonging to set $E(\tau) \backslash L^p(\tau)$ in slot $\tau$, then
\begin{align*}
    \max_{A \in \offalg} C_A( E(\tau) \backslash L^p(\tau) )_{\tau}    &= \sum_{i=1}^{k} v_i  \;- g(k),\\
    &\stackrel{(a)}{\leq} \sum_{i=1}^{k} v^{(i)}  \;- g(k) \; = V^{(k)} - g(k),\\
    &\leq \max_{k} \{ V^{(k)} - g(k) \} \stackrel{(b)}{=} C_{\text{Greedy}} (i^\star , \tau),
\end{align*}
where $(a)$ follows since sum of the values of any set of $k$ jobs is less than those of the $k$ highest valued jobs (where $v^{(i)}$ is the value of the $i^{th}$ highest valued job), and $(b)$ follows from the definition of $C_{\text{Greedy}}(i^\star,\tau)$ in Algorithm min-LCR since the jobs chosen by $A^\star$ belong to the set $E(\tau) \backslash L^p(\tau)$.
\end{proof}

Hence, using \eqref{eq:dummy3} and \eqref{eq:dummy2}, we have from \eqref{eq:int1}, 
\begin{align}\nn
    C_{\off}(\boldsymbol{\sigma}) &\le \sum_{\tau=1}^{\mathsf{last}} M_{i^\star,\tau} + \sum_{\tau=1}^{\mathsf{last}}C_{\greedy}(i^\star, \tau).
    \end{align}
    
From the definition of the Algorithm min-LCR, the profit made by it at slot $\tau$ is $P_{i^\star,\tau} = V^{(i^\star)} - g(i^\star)$ by processing $i^\star$ jobs at slot $\tau$. Thus, the competitive ratio of min-LCR Algorithm is 
at most 
\begin{align}\nn
    \sfr_{LCR}(\boldsymbol{\sigma}) = \frac{ C_{\off}(\boldsymbol{\sigma})}{C_{\text{LCR}}(\boldsymbol{\sigma})}& \le \frac{\sum_{\tau=1}^{\mathsf{last}} M_{i^\star,\tau} +C_{\greedy}(i^\star, \tau)}{\sum_{\tau=1}^{\mathsf{last}}P_{i^\star,\tau}}, \\ \label{eq:dummy51}
    &\le \max_\tau \frac{ M_{i^\star,\tau} +C_{\greedy}(i^\star, \tau)}{P_{i^\star,\tau}} = \max_\tau \lcr_{i^\star(\tau)}(\tau),
    \end{align}
    where $\mathsf{last}$ is the last slot at which min-LCR algorithm processes any jobs for the both the numerator and the denominator.

Next, to complete the proof of Theorem \ref{thm:lcr}, we show that the  optimal competitive ratio $\sfr^{\star}$ is at least $\max_\tau \lcr_{i^\star(\tau)}(\tau)$ for any input $\boldsymbol{\sigma}$ and any slot $\tau$.

\begin{lemma}\label{lem:lcrcomparison} For any input $\boldsymbol{\sigma}$,  $\max_{\tau} \lcr_{i^\star(\tau)}(\tau) \le \sfr^{\star}$.
\end{lemma}
\begin{proof} We will proceed via contradiction. Let the hypothesis $H_1 : \ \  \exists \  \tau,\boldsymbol{\sigma}_1 \ \text{such that}  \ \lcr_{i^\star(\tau)}(\tau) > \sfr^{\star}$. For ease of exposition in this proof, let $\lcr^{(\tau)}_{\; \boldsymbol{\sigma}_1} = \lcr_{i^\star(\tau)}(\tau)$.
Consider the set of jobs $E(\tau)$ at slot $\tau$ that is the union of all the non-processed and non-expired jobs till slot $\tau-1$ by the Algorithm min-LCR and the newly arrived jobs at slot $\tau$. Let $v_1, \dots, v_{|E(\tau)|}$ be the values of these jobs, where job $i$ arrived at slot $a_i$ and has deadline $d_i$.

We construct another input sequence $\boldsymbol{\sigma}_2$ that consists of 
$|E(\tau)|$ jobs with values $v_1, \dots, v_{|E(\tau)|}$. All the jobs of $\boldsymbol{\sigma}_2$ arrive at slot $1$, and the deadlines for each job is equal to the $d_i - \tau-a_i$. Important to note that $d_i$ is allowed to be arbitrary and unknown to the online algorithm while known to the $\off$ in both $\boldsymbol{\sigma}_1$ and $\boldsymbol{\sigma}_2$.

Consider a new hypothesis $ H_2 : \ \ \text{Optimum online algorithm $\opt$ on input 
$ \boldsymbol{\sigma}_2$}$
 has competitive ratio lower than $\lcr^{(1)}_{\; \boldsymbol{\sigma}_2}$.
It is easy to check that from the min-LCR algorithm definition, that $\lcr^{(1)}_{\; \boldsymbol{\sigma}_2}=\lcr^{(\tau)}_{\; \boldsymbol{\sigma}_1}$. 
From hypothesis $H_1$, 
$\lcr^{(\tau)}_{\; \boldsymbol{\sigma}_1} > \sfr^\star$, which implies that there exists an optimum online algorithm $\opt$ that on input $ \boldsymbol{\sigma}_2$ has competitive ratio lower than $\lcr^{(1)}_{\; \boldsymbol{\sigma}_2}$, since $\sfr^\star$ is achievable. Hence $H_1 \implies H_2$. Now we will contradict hypothesis $H_2$.

Let the optimal online algorithm $\opt$ on input $\boldsymbol{\sigma}_2$ process any $1\le k\le |E(\tau)| $ jobs at slot $1$, leaving the remaining jobs for later slots. Since the deadlines $d_i$ are arbitrary, let the true deadlines for the $k$ jobs (whatever the choice of $k$ may be for $\opt$) that were processed by $\opt$ in slot $1$ to be $\infty$, while keeping the deadlines for all jobs other than $k$ selected ones to be slot $1$ itself. Thus, there are no available jobs at slot $2$ for $\opt$. 
Thus, the $\opt$ can make the maximum profit by sending the $k$ highest valued jobs of $\boldsymbol{\sigma}_2$ or $E(\tau)$ in slot $1$.

Given that offline optimal $\off$ knows the deadlines, $\off$ processes the $k$ jobs chosen by 
$\opt$ for processing in slot $1$ in $k$ slots starting from the second slot individually, and among the rest of 
$ |E(\tau)| -  k$ jobs processes as many jobs in slot $1$ to maximize its profit in slot $1$, which by definition is $C_{\greedy}(|E(\tau)| -  k,1)$. Thus, $C_{\off}(\boldsymbol{\sigma}_2) = M_{k,1} + C_{\greedy}(|E(\tau)| -  k,1)$, while  $C_{\opt}(\boldsymbol{\sigma}_2)= \sum_{j=1}^{j=k} v^{(j)} \;\; - g(k) = P_{k,1}$. Therefore,
\begin{align}\nn
  \dfrac{C_{\off}(\boldsymbol{\sigma}_2)}{C_{\opt}(\boldsymbol{\sigma}_2)} &= \dfrac{M_{k,1} + C_{\greedy}(|E(\tau)| -  k,1)}{P_{k,1}}= \lcr_{k}(1),\\
    &\geq \lcr_{i^\star(1)}(1).
\end{align}
since $\lcr_{i^\star(1)}(1)$ is the minimum $\lcr$ over all possible $k$. Thus, the optimum online algorithm $\opt$ on input $ \boldsymbol{\sigma}_2$
 cannot have a competitive ratio lower than $\lcr_{i^\star(1)}(1) = \lcr^{(1)}_{\; \boldsymbol{\sigma}_2}$.
Thus, we get  contradiction to hypothesis $H_2$ and equivalently to $H_1$.
\end{proof}
%\begin{remark}If suppose the function $g(k)$ is concave, e.g., $g(k) = log(k+1)$, then the effective cost of processing the $i^{th}$ job $c_i= g(i)-g(i-1)$ decreases with $i$. Thus, the min-LCR algorithm when $g(k)$ is concave can have arbitrarily large competitive ratio as follows. Consider the example, when there are only two jobs with values $0.6$ each, with deadline as slot $1$ itself for both of them. In this case, the min-LCR algorithm computes $m$ to be zero, since $0.6 - g(1) < 0$, while the $\off$ processes both jobs and accrues a profit of $(0.6-g(1)) + (0.6 - g(2)+g(1)) = 0.1$ for $g(k) = log(k+1)$. Thus, for concave function $g$, the min-LCR algorithm provides a degenerate bound on the competitive ratio.
%\end{remark}
\begin{remark}The typical methods to show lower bounds for online algorithms are: either by explicit construction of a lower bound example, or using Yao's minimax Lemma for randomized algorithms. 
The procedure in this paper is different, however, results in the same conclusion, that the competitive ratio of the proposed min-LCR algorithm is as good as any optimal online algorithm. The novelty is that the technique surprisingly works for any general convex energy cost function, and without having to explicitly evaluate the lower or upper bounds, which is very rarely found in literature.
\end{remark}
After establishing that the min-LCR Algorithm is an optimal online algorithm for any convex cost function $g$, we next concentrate on a specific cost function $g(k) = k^\alpha$ for $\alpha\ge 2$ that is the most popular choice in literature, to derive some concrete bounds on the competitive ratio of the min-LCR Algorithm. Recall that min-LCR algorithm involves finding the number of jobs $k$ that minimizes the $\lcr_k$, which can be exhaustive. We next propose a simplified version of the min-LCR algorithm 
where exactly $k=\left \lfloor{\beta m}\right \rfloor$ or $k=\left \lceil{\beta m}\right \rceil$ jobs are processed in each slot depending on whichever one has lower LCR, and where $\beta$ is the solution of the equation $x^\alpha + x^{\alpha-1} -1=0$. We also consider the natural greedy special case of min-LCR algorithm that always processes $m$ jobs, where $m$ has been defined in the min-LCR algorithm.

We first derive a lower bound on the competitive ratio of any online algorithm, and next show that the simplified min-LCR Algorithm (and consequently the min-LCR algorithm) achieves that lower bound for $\alpha =2$. For $\alpha >2 $, we show that the competitive ratio of the simplified min-LCR Algorithm is at most $3$, while for $\alpha \ge 2.5$ it is at most $2.618$.
\section{$g(k) = k^\alpha$ for $\alpha\ge 2$}
\subsection{Lower Bound on the Competitive Ratio for $\alpha= 2$}
\begin{lemma}\label{lem:lb2} For any online algorithm $\alg$,  $\sfr_{\alg}  \ge \phi + 1$ for cost function $g(k) = k^\alpha$ with $\alpha= 2$, where $\phi = 1/\delta$ and $\delta=\frac{\sqrt{5}-1}{2}$.\end{lemma}
\begin{proof}
We consider $2z$ jobs each with value $2z$ arriving at the start of slot $1$. Thus, in slot $1$ at most $z$ jobs can be processed by any algorithm since the cost function is $g(k) = k^2$, and the effective cost of processing job $z+1$ or higher is at least $g(z+1) - g(z) > 2z$. Let an online algorithm $\alg$ choose to process $k$ jobs out of the maximum possible $z$. Then we adversarially choose  the deadlines of these $k$ jobs to be infinity $\infty$, while keeping the deadlines of all other remaining jobs to be slot $1$ itself. 
Since the offline optimal $\off$ knows the deadlines, it processes the maximum possible $m$ jobs in slot $1$, while processes the $k$ jobs chosen by the $\alg$ for slot $1$, one at a time starting from slot $2$ in $k$ slots, making a profit of $2z \cdot z- z^2 + 2z \cdot k -k$, while the $\alg$ makes only a profit of $2zk-k^2$. Thus, the competitive ratio of any online algorithm $\alg$ as a function of $k$ is 
$$\sfr_{\alg}  \ge \min_k \frac{2z \cdot z - z^2 + 2z \cdot k -k}{2zk-k^2}.$$
We take the limit as $z \rightarrow \infty$ to get that 
$$\sfr_{\alg}  \ge \lim_{z\rightarrow \infty}\min_k \frac{2z \cdot z - z^2 + 2z \cdot k -k}{2zk-k^2} = \lim_{z\rightarrow \infty}\min_k \frac{2z \cdot z - z^2 + 2z \cdot k}{2zk-k^2} .$$ since $\lim_{z\rightarrow \infty} \frac{-k}{2zk-k^2} = 0$ for any choice of $k$ including the optimizer. It is easy to show that for $\delta=\frac{\sqrt{5}-1}{2}$,  (see Appendix \ref{sec:mincran})
\begin{equation}\label{eq:dummy71}
\delta z = \arg \min_k \frac{2z \cdot z - z^2 + 2z \cdot k}{2zk-k^2},
\end{equation}
 and $\lim_{m\rightarrow \infty} \frac{2z \cdot z - z^2 + 2z \cdot k}{2zk-k^2}\vert_{k=\frac{\sqrt{5}-1}{2}z} = \phi+1$, 
where $\phi = \frac{1}{\delta}$.
\end{proof}
One can proceed similarly and get an expression for the lower bound on the competitive ratio for all values of $\alpha>2$ as 
\begin{equation}\label{eq:dummy100}\max_{z \in \mathbb{Z^+}} \max_{x \le g(z+1) - 2 g(z)+g(z-1)} \min_{1\leq k\leq z} \bigg \{ \dfrac{[k(z^\alpha -(z-1)^\alpha +x -1)] + [z(z^\alpha - (z-1)^\alpha +x) - z^\alpha]}{[k(z^\alpha -(z-1)^\alpha +x) - k^\alpha] } \bigg \},
\end{equation}
however, it is not easy to simplify it analytically, needing a numerical solution as presented in Fig. \ref{fig:lb}. To be more concrete, we next derive a slightly loose lower bound for $\alpha >2$ as follows.

\subsection{Lower Bound on the Competitive Ratio for $\alpha> 2$}
\begin{lemma} For any online algorithm $\alg$,  $\sfr_{\alg}  \ge \sqrt{2} + 1$ for cost function $g(k) = k^\alpha$ for all $\alpha> 2$.
\end{lemma}
\begin{proof} Consider the input where four jobs arrive at the beginning of slot $1$ each with value $v = \Big ( 1 + \frac{1}{\sqrt{2}} \Big)[g(2) - \sqrt{2}g(1)] $. Any online algorithm $\alg$ can process $k=1$ or $2$ or $3$ or $4$ jobs in slot $1$. The choice of 
$v$ is such that $v < g(3)- g(2)$, where $g(3)- g(2)$ is the effective  cost of processing the third job in slot $1$. Moreover, since $g$ is a convex function $g(4)- g(3)> g(3)- g(2)$. Therefore processing the third or the fourth job in slot $1$ incurs negative profit, and hence 
at most $2$ jobs can be processed in slot $1$. 

Now depending on the choice of $\alg$ for $k=1,2$, the adversarial choice of deadline will be that those $k$ jobs will have deadline $\infty$, while the remaining $4-k$ jobs will have deadline as slot $1$ itself. A simple enumerative exercise reveals that 
$\sfr_{\alg}\vert_{k=1} = \sqrt{2} + 1, \sfr_{\alg}\vert_{k=2} = \sqrt{2} + 1$. The choice of $v$ is the only non-trivial part in deriving this lower bound.
\end{proof}
\subsection{Upper Bound on the Competitive Ratio for $\alpha\ge 2$}

From here onwards we derive upper bounds on the competitive ratio of a simplified min-LCR Algorithm, where we choose a particular number of jobs to process in every slot. The simplified min-LCR Algorithm is as follows. 
\begin{algorithm}[ht]
%\label{algorithm: Simplified LCR}
\caption{Simplified min-LCR Algorithm (sim-LCR)}
\SetKwInOut{Input}{Input}
\Input{Sequence $\boldsymbol{\sigma} = J_1, \dots, J_n$}
At slot $\tau$, consider the union of all the non-processed jobs by the algorithm so far that have not-expired and the newly arriving jobs in slot $\tau$ and call it $E(\tau)$ \\
Arrange the jobs in $E(\tau)$ in non-increasing order of their payoffs/value $v$ \\
Let $E(\tau)^i = \{\text{first} \ \ i \ \text{jobs in} \ E(\tau)\}$ and $E^{i-}(\tau) = E(\tau) \backslash E(\tau)^i$\\
Let $v^{(i)}$ be the value of job of $E(\tau)$ with the $i^{th}$ highest value, $V^{(i)}$ be the sum of the values of the first $i$ jobs with highest values\\
Let $m =  \max_{\{v^{(j)} - [g(j)-g(j-1)] > 0\}} j $ \\
Let $M_{i,\tau} = V^{(i)} - i g(1),  P_{i,\tau} = V^{(i)} - g(i), C_{\greedy}(i,\tau) = \max_{j \in E^{i-}(\tau) } \{ {V^{(j)} - g(j)} \} $\\
		Let $\beta$ be the solution of $x^\alpha + x^{\alpha-1} -1=0$ \\ 
		Compute $\lcr_{i}(\tau) = \dfrac{M_{i,\tau} + C_{\greedy}(i,\tau)}{P_{i,\tau}}$ as in min-LCR algorithm for $i=\left \lfloor{\beta m}\right \rfloor$ or $i=\left \lceil{\beta m}\right \rceil$\\
		Process either $i=\left \lfloor{\beta m}\right \rfloor$ or $i=\left \lceil{\beta m}\right \rceil$ jobs whichever one has lower $\lcr_{i}$ in slot $\tau$.
\end{algorithm}

\begin{lemma}\label{lem:compratiolcrub1} The competitive ratio of sim-LCR Algorithm is at most $\phi+1$ for $\alpha=2$ or $\alpha \ge 2.5$.
\end{lemma}
Combining Lemma \ref{lem:lb2} and Lemma \ref{lem:compratiolcrub1}, we get the following result.
\begin{corollary}\label{cor:compratiolcrub1} sim-LCR Algorithm is an optimal online algorithm for $\alpha=2$ and the optimal competitive ratio for $\alpha=2$ is $\phi+1$.
\end{corollary}
Proof of Lemma \ref{lem:compratiolcrub1} is as follows.
\begin{proof}
Let $i(\tau)$ be the optimizer among $i=\left \lfloor{\beta m}\right \rfloor$ or $i=\left \lceil{\beta m}\right \rceil$ that minimizes the $\lcr_i$ with the sim-LCR algorithm at slot $\tau$. Then following an identical proof as for Lemma \ref{lem:compratiolcrub}, it follows that the competitive ratio of the sim-LCR algorithm from \eqref{eq:dummy51}, 

\begin{align}\nn
    \sfr_{\text{sim-LCR}}(\boldsymbol{\sigma}) = \frac{ C_{\off}(\boldsymbol{\sigma})}{C_{\text{sim-LCR}}(\boldsymbol{\sigma})}& \le \frac{\sum_{\tau=1}^{\mathsf{last}} M_{i(\tau),\tau} +C_{\greedy}(i(\tau), \tau)}{\sum_{\tau=1}^{\mathsf{last}}P_{i(\tau),\tau}} \\ \label{eq:dummy61}
    &\le \max_\tau \frac{ M_{i(\tau),\tau} +C_{\greedy}(i(\tau), \tau)}{P_{i(\tau),\tau}} = \max_\tau \lcr_{i(\tau)}(\tau).    \end{align}

Thus, to complete the proof, in the following, we show that the $\lcr_k$ for either $k=\left \lfloor{\beta m}\right \rfloor$ or $k=\left \lceil{\beta m}\right \rceil$ is less than $\phi+1$ for the sim-LCR Algorithm at all slots $\tau$.
From the definition of the $\lcr$
\begin{align}\nn
\lcr_k &= \dfrac{M_{k,\tau} + C_{\greedy}(k, \tau)}{P_{k,\tau}},\\  \label{eq:dummy31}
&\stackrel{(a)}\leq \dfrac{\sum_{j=1}^{j=k} v^{(j)} \; - \;k\cdot g(1) + \sum_{j=1}^{j=m} v^{(j)} \; - \;g(m)}{\sum_{j=1}^{j=k} \;v^{(j)} \; - \;g(k)},
\end{align}
$(a)$ follows from the definition of $M_{k,\tau}$  and an upper bound derived for $C_{\greedy}(k, \tau)$ in Appendix \ref{sec:greedyub}.
Recall that the set of available jobs $E(\tau)$ at each slot $\tau$ are arranged in decreasing order of their values, hence, for  $k\leq m$, the average of the values of the first $k$ jobs is greater than the average of the values of the first $m$. Hence, we have $\frac{m}{k}\sum_{j=1}^{j=k} v^{(j)} \geq \sum_{j=1}^{j=m} v^{(j)}$.
Substituting this in \eqref{eq:dummy31} yields
\begin{align}\label{eq:dummy32}
 \lcr_k &\leq \dfrac{\sum_{j=1}^{j=k} v^{(j)} \; - \;k + \frac{m}{k}\big( \sum_{j=1}^{j=k} v^{(j)} \big) \; - \;g(m)}{\sum_{j=1}^{j=k} \;v^{(j)} \; - \;g(k)} = \dfrac{m}{k} +1 + \dfrac{g(k) + \frac{m}{k}g(k) - g(m) -k}{\sum_{j=1}^{j=k} \;v^{(j)} \; - \;g(k)}.
\end{align}
Consider the numerator of the second term $f(k) = g(k) + \frac{m}{k}g(k) - g(m) -k$ at $k=\beta m$. By definition, $f(\beta m) = m^\alpha(\beta^\alpha+\beta^{\alpha-1}-1) -\beta m\le  0$ since $\beta$ is the solution of the equation $x^\alpha+x^{\alpha-1}-1=0$. Moreover, $f(\lfloor\beta m\rfloor) =  
(\lfloor\beta m\rfloor)^\alpha + \frac{m}{(\lfloor\beta m\rfloor)^\alpha}(\lfloor\beta m\rfloor)^\alpha - m^\alpha 
 -\lfloor\beta m\rfloor\le 0$, since 
$(\lfloor\beta m\rfloor)^\alpha + \frac{m}{(\lfloor\beta m\rfloor)^\alpha}(\lfloor\beta m\rfloor)^\alpha - m^\alpha \le m^\alpha(\beta^\alpha+\beta^{\alpha-1}-1)$. Using this in \eqref{eq:dummy32}, we get 
\begin{align}\label{eq:dummy34}
 \lcr_{\lfloor\beta m\rfloor} &\leq \dfrac{m}{k} +1.
\end{align}
For $\alpha \geq 2.5$, by direct computation, one can check that for $m=\{ 1,3,6,8,9,10,11,12 \}$ and $\forall m \geq 13, \lcr_{\lfloor\beta m\rfloor} \leq \phi +1$.
For the remaining values of $m=2,4,5,7 \; \text{ either } \lcr_{\lfloor\beta m\rfloor} \leq \phi +1$ or $\lcr_{\lceil\beta m\rceil} \leq \phi +1$, as derived in Appendix \ref{sec:laborcompm}.

For $\alpha=2$, a little more involved approach is needed as follows. 
For $\alpha =2$, $\beta = \delta$, where $\beta$ is the solution of the equation $x^\alpha+x^{\alpha-1}-1=0$. Let 
$\gamma$ be the positive root of the equation $x^2 + (m-1)x -m^2 = 0$. From \eqref{eq:dummy31}, using $\alpha=2$,
\begin{equation}\label{eq:dummy41}\lcr_k= \dfrac{m}{k} +1 + \dfrac{k^2 + mk - m^2  -k}{\sum_{j=1}^{j=k} \;v^{(j)} \; - \;k^2}.
\end{equation} 
The quadratic equation $k^2 + mk - m^2  -k$ is increasing in the interval $[\delta m, \gamma]$ ($\delta = \frac{\sqrt{5}-1}{2}$) with value $<0$ at $k = \delta m$ and $=0$ at $k=\gamma$ (by definition of $\gamma$). 
So at all intermediate values of $k \in (\delta m, \gamma)$, the value of the equation must be less than $0$. Thus, if $\left \lceil{\delta m}\right \rceil \in [\delta m, \gamma]$, we have 
$ \lcr_{\left \lceil{\delta m}\right \rceil} \leq \dfrac{m}{ \left \lceil{\delta m}\right \rceil} + 1 \leq \dfrac{m}{\delta m } + 1= \phi + 1$.

Next, we consider the case when $\left \lceil{\delta m}\right \rceil \in (\gamma, \delta m + 1]$. The numerator of the third term in the RHS of \eqref{eq:dummy41} is an increasing function for $(\gamma, \delta m + 1]$ from the definition of $\gamma$. Thus, since $\sum_{j=1}^{j=k} \;v^{(j)} \geq k[g(m) - g(m-1)]$ where $v^{(j)}$ is the $j^{th}$ highest value of the job when arranged in a non-increasing order of values.
Thus, from \eqref{eq:dummy41}, we get 
\begin{equation}\label{eq:dummy42}\lcr_k \le \dfrac{m}{k} +1 + \dfrac{k^2 + mk - m^2  -k}{(2m-1)k -k^2}.
\end{equation}
Let $\psi(k) = \dfrac{m}{k} +1 + \dfrac{k^2 + mk - m^2  -k}{(2m-1)k -k^2}$, then it follows that 
$\dfrac{\partial \psi(k)}{\partial k} \ge 0$ for $k \in (\gamma, \delta m + 1]$, and 
$\psi(\delta m + 1) \le \phi+1$ for $m\ge 5$. This implies that $\lcr_k\le \phi+1$ for all $k \in (\gamma, \delta m + 1]$. Thus, even if $\left \lceil{\delta m}\right \rceil \in (\gamma, \delta m + 1]$, 
$ \lcr_{\left \lceil{\beta m}\right \rceil} \le \phi+1$ for $m\ge 5$. For $m=1$, $\lcr_{\left \lceil{\delta m}\right \rceil} \le 2$, while for  
\begin{align}
m=3 &\implies LCR_{\left \lceil{\delta m}\right \rceil} \leq \dfrac{3}{2} + 1+\dfrac{2^2 + 3\cdot2 -3^2 -2}{(2\cdot3 -1 )2 - 2^2}  < \phi+1\\
m=4 &\implies LCR_{\left \lceil{\delta m}\right \rceil} \leq \dfrac{4}{3} + 1 + \dfrac{3^2 + 4\cdot 3 - 4^2 -3}{(2 \cdot 4 -1 )3 -3^2}  < \phi+1.
\end{align}

%$\lcr_{\left \lceil{\delta m}\right \rceil} \leq \dfrac{3^2 + 2\cdot3\cdot2 -3 -2\cdot2}{2\cdot3\cdot2 - 3^2 - 4}  < \phi+1$, and for $m=4$, $\lcr_{\left \lceil{\delta m}\right \rceil} \leq \dfrac{4^2 + 2\cdot4\cdot3 - 4-2\cdot3}{2\cdot4\cdot3 - 3^2 - 3}  < \phi+1$.
For $m=2$, the proof is identical to the one for the case of $\alpha \ge 2.5$ for $m=2$ as provided in Appendix \ref{sec:laborcompm}.
\end{proof}

\subsection{Upper Bound on the Competitive Ratio of the min-LCR Algorithm for $\alpha\ge 2$}
We can obtain a slightly loose upper bound for all values of $\alpha\ge 2$ by enforcing the min-LCR Algorithm to process all $m$ jobs, where $m$ is largest number of jobs that can be processed at any slot with positive profit for each job. This restriction can essentially be seen as a greedy analogue of the min-LCR Algorithm, and we call it $\greedy$. 
\begin{lemma}\label{lem:greedy}The competitive ratio of the $\greedy$ algorithm (min-LCR Algorithm with $i^\star =m$) is at most $3$ for all $\alpha >2$.
\end{lemma}
Proof can be found in Appendix \ref{app:greedy}.

\newpage
\bibliographystyle{ACM-Reference-Format}
\bibliography{refs}
\newpage
%%%%%%
\section{Appendix}
\subsection{Numerical Evaluation of the Lower bound \eqref{eq:dummy100}}
\begin{figure}
\centering
\includegraphics[width=3.5in]{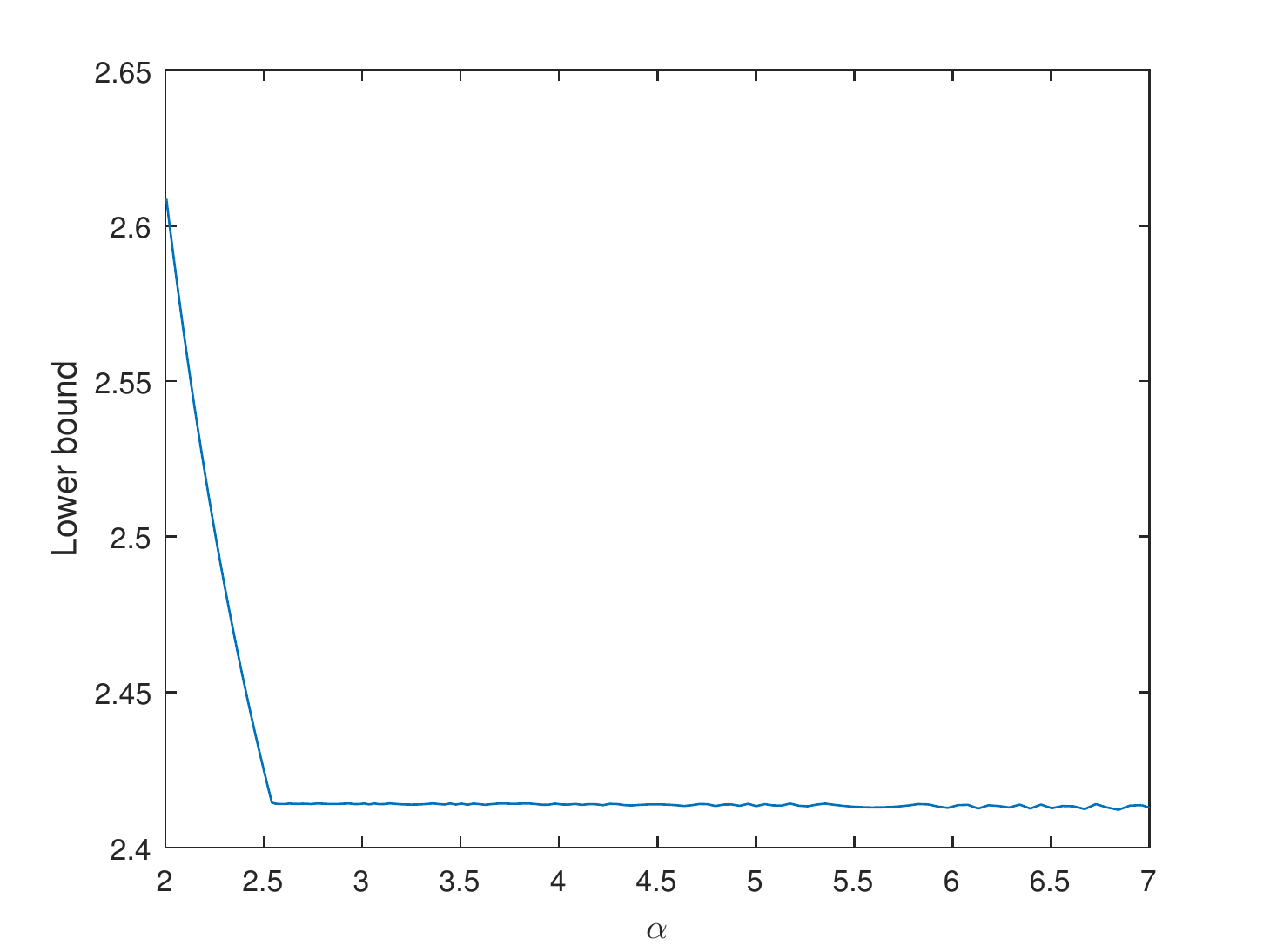}
\caption{Numerical Evaluation of the Lower bound \eqref{eq:dummy100} as a function of $\alpha$.}
\label{fig:lb}
\end{figure}
\subsection{Proof for upper bounding $C_{\greedy}(k,\tau)$ in \eqref{eq:dummy31}}\label{sec:greedyub}
\begin{lemma}\label{lem:greedyub}
$$C_{\greedy}(k,\tau) \leq \sum_{j=1}^{j=m} v^{(j)} - g(m).$$
\end{lemma}
\begin{proof} Recall that the effective cost of the $i^{th}$ job is $c_i = g(i)-g(i-1)$.
\begin{align*}
    C_{\greedy}(k,\tau) &\stackrel{(a)}{=} \max_{j \in E^{k-}(\tau)} (V^{(j)} - g(j)),\\
    &\stackrel{(b)}\leq \max_{j \in E(\tau)} (V^{(j)} - g(j)),\\
    &\stackrel{(c)}= (V^{(j^\star)} - g(j^\star)) = \sum_{i=1}^{j^\star} v^{(i)} - c_i,\\
    &\stackrel{(d)}{=} \sum_{j=1}^{j=m} v^{(j)} - g(m),
\end{align*}
$(a)$ follows from the definition, (b) follows since $E^{k-}(\tau) \subseteq E(\tau)$, $(c)$ follows by defining the optimizer $j$ in $(b)$ to be $j^\star$, while $(d)$ is true because of the following inequalities
$$\forall i > m : \;\; v^{(i)} \stackrel{c}{\leq} v^{(m+1)} \stackrel{d}{\leq} c_{m+1} \stackrel{e}{\leq} c_i \implies [v^{(i)} - c_i] \leq 0$$
$$\forall i \leq m : \;\; v^{(i)} \stackrel{c}{\geq} v^{(m)} \stackrel{d}{\geq} c_m \stackrel{e}{\geq} c_i \implies [v^{(i)} - c_i] \geq 0,$$
where $(c)$ follows since jobs are indexed in the non-increasing order of their values, and $(d)$ follows from the definition of $m$, and $(e)$ follows from the convexity of $g(k)$.   
\end{proof}

\subsection{Proof of \eqref{eq:dummy71}}\label{sec:mincran}
\begin{lemma}
$$\arg \min _{k} \bigg\{ \dfrac{z^2 + 2zk}{2zk - k^2} \bigg \} = \delta z.$$
\end{lemma}

\begin{proof}
Taking the derivative, we get  
\begin{align}
    \dfrac{\partial }{\partial k}\;\bigg ( \dfrac{2zk + z^2}{2zk-k^2} \bigg ) = \dfrac{2z(k^2 +zk -z^2)}{(2zk-k^2)^2},
\end{align}
which when equated to zero, we get $k^\star = \delta z$.
Taking the second derivative and evaluating at $\delta m$, we show that $k^\star = \delta m$  is a local minima as follows.
\begin{align*}
    \dfrac{\partial ^2 }{\partial k^2}\;\bigg ( \dfrac{2zk + z^2}{2zk-k^2} \bigg ) &= \dfrac{[2z (2zk-k^2)]\;[2z^3 -6z^2k+4z^3+3zk^2]}{(2zk-k^2)^4},\\
    \implies \dfrac{\partial ^2 }{\partial k^2}\;\bigg ( \dfrac{2zk + z^2}{2zk-k^2} \bigg ) \bigg \vert _{k=\delta z} &\stackrel{(a)}{>} 0,
\end{align*}
where $(a)$ follows since $$[2k^3 -6z^2k+4z^3+3zk^2] \Big \vert _{k= \delta z} = [2z^3(2\delta^3 +3\delta^2 -6\delta +4)] > 0.$$
Evaluating at the boundary points of $k=1$, we get $\dfrac{z^2 +2z}{2z -1} > \phi+1$  and at $k=z$, $\dfrac{z^2 +2z^2}{2z^2 -z^2} = 3 \geq \phi+1$.
Thus, we conclude that $k=\delta z$ is the minima.\\
\end{proof}

\subsection{Proof of bounding $ \lcr_{\lfloor\beta m\rfloor}$ and $\lcr_{\lceil\beta m\rceil}$ for $m=2,4,5,7$} \label{sec:laborcompm}

\begin{lemma}
For $m=\{2,4,5,7\} \; \forall \alpha \geq 2.5 \;$ either $\; \lcr_{\left \lfloor{\beta m}\right \rfloor} $ or $\lcr_{\left \lceil{\beta m}\right \rceil}$ is $\leq \phi+1$.
\end{lemma}

\begin{proof} $m=2$: 
$\alpha \geq 2 \implies {\left \lfloor{\beta m}\right \rfloor} = 1 \text{  and  } \left \lceil{\beta m}\right \rceil = 2.$
When $v^{(2)} \leq \Big (1+ \frac{1}{\sqrt{2}} \Big)[g(2) -\sqrt{2} g(1)]$, 
$$\lcr_1 \stackrel{p}{=} \dfrac{[v^{(1)} - g(1)] + [v^{(2)} + v^{(3)} -g(2)] }{ [v^{(1)} - g(1)] } \stackrel{q}{\leq} 1 + \dfrac{2v^{(2)} - g(2)}{v^{(2)} -g(1)} \stackrel{r}{\leq} \phi+1,$$ while when $v^{(2)} \geq \Big (1+ \frac{1}{\sqrt{2}} \Big)[g(2) -\sqrt{2} g(1)]$, 
$$\lcr_2 \stackrel{p}{=} \dfrac{[v^{(1)}+v^{(2)}-2g(1)] + [v^{(3)}+v^{(4)} -g(2)]}{[v^{(1)}+v^{(2)}-g(2)]} \stackrel{q}{\leq} 1 + \dfrac{2v^{(2)} -2g(1)}{2v^{(2)} - g(2)} \stackrel{r}{\leq} \phi+1,$$
where $p$ follows from the definition of $\lcr$, $q$ follows from the fact that $v^{(1)} \geq v^{(2)} \geq v^{(3)} \geq v^{(4)}$ and $r$ follows by simple evaluation in the appropriate range of choice for $v^{(2)}$.
$m=4$: 
When $\alpha \in [2.5,2.945] \implies \left \lceil{\beta m}\right \rceil = \left \lceil{\beta * 4}\right \rceil = 3 $\\
$$
   \implies \lcr_{\left \lceil{\beta m}\right \rceil} 
    \stackrel{a}{=} \dfrac{4}{3} + 1 + \dfrac{7\cdot 3^{(\alpha-1)} - 4 \cdot 4^{(\alpha-1)}}{12(4^{(\alpha-1)} - 3^{(\alpha-1)})}\\
    \stackrel{b}{\leq} 2 +  \dfrac{1}{4 \Big(\frac{4^{(\alpha-1)}}{3^{(\alpha-1)}} -1 \Big )}\\
    \stackrel{c}{\leq} \phi+1,
$$
$$\alpha \geq 2.945 \implies \left \lfloor{\beta m}\right \rfloor \geq 3 \implies \lcr_{\left \lfloor{\beta m}\right \rfloor} \stackrel{d}{\leq} \dfrac{m}{k} +1 = \frac{4}{3}+1 \leq \phi+1.$$\\\\
$m=5$:
When $\alpha \in [2.5,3.641] \implies \left \lceil{\beta m}\right \rceil = 4$ \\
$$ \implies
    \lcr_{\left \lceil{\beta m}\right \rceil} \stackrel{a}{\leq} \dfrac{5}{4} + 1 + \dfrac{9\cdot 4^{(\alpha-1)} - 5 \cdot 5^{(\alpha-1)}}{20(5^{(\alpha-1)} - 4^{(\alpha-1)})}
    \stackrel{b}{=} 2 + \dfrac{1}{5 \Big ( \frac{5^{(\alpha-1)}}{4^{(\alpha-1)}} -1 \Big )}
    \stackrel{e}{\leq} \phi+1,$$
$$ \alpha \geq 3.641 \implies \left \lfloor{\beta m}\right \rfloor \geq 4 \implies \lcr_{\left \lfloor{\beta m}\right \rfloor} \stackrel{d}{\leq} \dfrac{m}{k} +1 = \frac{5}{4}+1 \leq \phi+1. $$\\\\
$m=7$:
$$\alpha \in [2.5,2.6] \implies \left \lceil{\beta m}\right \rceil = 5 \implies \lcr_{\left \lceil{\beta m}\right \rceil} \stackrel{a}{\leq} \; \dfrac{7}{5} + 1 + \dfrac{\frac{12}{5}5^{\alpha} - 7^{\alpha} -5 }{ 5[7^{\alpha} - 6^{\alpha} - 5^{\alpha}] } \leq  \; \phi+1,$$
$$ \alpha \geq 2.6 \implies \left \lfloor{\beta m}\right \rfloor \geq 5 \implies \lcr_{\left \lfloor{\beta m}\right \rfloor} \stackrel{d}{\leq} \dfrac{m}{k} +1 = \frac{7}{5}+1 \leq \phi+1.$$\\\\

In all the cases, $a$ follows from \eqref{eq:dummy32} along with the fact that $\sum_{j=1}^{j=k} \;v^{(j)} \geq k[g(m) -g(m-1)]$ and $g(k) = k^{\alpha}$, $(b)$ follows by re-arranging terms, $c$ is true because 
$\forall \; \alpha \geq 2.5 \;$ we have $ \Big(\frac{4^{(\alpha-1)}}{3^{(\alpha-1)}} -1 \Big ) -1 \geq 0.53$, $d$ follows from \eqref{eq:dummy34}, and $e$ is true because $\forall \alpha \geq 2.5$ we have $ \Big ( \frac{5^{(\alpha-1)}}{4^{(\alpha-1)}} -1 \Big ) \geq 0.39$.
\end{proof}

\subsection{Proof of Lemma \ref{lem:greedy}}\label{app:greedy}
\begin{proof}
From \eqref{eq:dummy61}, $\sfr_{\text{sim-LCR}}(\boldsymbol{\sigma}) \le \max_\tau \lcr_{i(\tau)}$. In this greedy case, $i(\tau) = m$, always, and hence to prove the result we show that $ \lcr_{m} \le 3$ for all slot $\tau$ and all inputs $\boldsymbol{\sigma}$.
From \eqref{eq:dummy31}, we have that
$\lcr_k \leq \dfrac{m}{k} +1 + \dfrac{g(k) + \frac{m}{k}g(k) - g(m) -k}{\sum_{j=1}^{j=k} \;v^{(j)} \; - \;g(k)}$.
Choosing $k=m$ and substituting it in the above equation, 
\begin{align*}
    \lcr_m &\leq \dfrac{m}{m} +1 + \dfrac{g(m) + \frac{m}{m}g(m) - g(m) -m}{\sum_{j=1}^{j=m} \;v^{(j)} \; - \;g(m)},\\
    &\leq 1 + 1 + \dfrac{g(m)-m}{  \sum_{j=1}^{j=m} \;v^{(j)} \; - \;g(m) },\\
    &\stackrel{(a)}{=} 1 + \dfrac{  \sum_{j=1}^{j=m} \;v^{(j)} \; - \;m }{\sum_{j=1}^{j=m} \;v^{(j)} \; - \;g(m)},\\
    &= 1 + \dfrac{1}{\frac{\sum_{j=1}^{j=m} \;v^{(j)} \; - \;g(m)}{\sum_{j=1}^{j=m} \;v^{(j)} \; - \;m}},\\
    &= 1 + \dfrac{1}{1- \; \frac{g(m) - m}{\sum_{j=1}^{j=m} \;v^{(j)} \; - \;m}},\\
    &\stackrel{(b)}{=} 1 + \dfrac{1}{1-h(m,g)},
\end{align*}
where $(a)$ follows by adding the second and third terms, and $(b)$ by defining of $h(m,g)= \frac{g(m) - m}{\sum_{j=1}^{j=m} \;v^{(j)} \; - \;m}$. 

Next, we prove that $h(m,g) \le 1/2$, for $\alpha \ge 2$, from which it follows that 
$\lcr_m \le 3$. Thus, the competitive ratio of the min-LCR Algorithm is at most $3$ for all $\alpha >2$.
%Recall that $h(m,g) = \dfrac{g(m) - m}{\sum_{j=1}^{m} v^{(j)} -\; m}.$

From the definition of $m$, $m =  \max_{\{v^{(j)} - [g(j)-g(j-1)] > 0\}} j $,  all the $m$ highest valued jobs have a value greater than or equal to the effective cost of processing them in slot $m$, i.e., $v_k \ge c_k= v(k)- v(k-1)$ for $k\le m$,
and hence $\sum_{j=1}^{m} v^{(j)} \geq m [g(m) - g(m-1)]$, thus, 
$$ h(m,g) \leq \dfrac{g(m) - m}{m [g(m) - g(m-1)] -\; m}.$$\\
Let $\Theta(\alpha) =  \dfrac{g(m) - m}{m [g(m) - g(m-1)] -\; m} = \dfrac{m^{\alpha} - m}{m [m^{\alpha} - (m-1)^{\alpha}] -\; m}$, since $g(k) = k^{\alpha}$.

\begin{lemma}
$\forall \alpha \geq 2 \;\;\;\; \dfrac{\partial \Theta(\alpha)}{\partial \alpha} \leq 0.$
\end{lemma}

\begin{proof}
Taking the derivative and rearranging the terms we get 
\begin{align*}
    \dfrac{\partial \Theta(\alpha)}{\partial \alpha} &= -\dfrac{  m^2 (m-1)[ m^{\alpha -1}((m-1)^{\alpha -1}-1)\log(m) - (m-1)^{\alpha-1}(m^{\alpha -1} -1) \log(m-1)  ]   }{ (m [m^{\alpha} - (m-1)^{\alpha}] -\; m)^2  },\\
    &\stackrel{(a)}{=} - \bigg[ \dfrac{(m^2)(m-1)((m-1)^{\alpha -1}-1)(m^{\alpha -1} -1)}{  (\alpha -1)(m [m^{\alpha} - (m-1)^{\alpha}] -\; m)^2 } \bigg] \\
    & \ \ \ \ \ \ \ \ \cdot \bigg[  \dfrac{ (m^{\alpha -1})\log(m^{\alpha -1}) }{ (m^{\alpha -1} -1) }  - \dfrac{ ((m-1)^{\alpha -1})\log((m-1)^{\alpha -1}) }{ ((m-1)^{\alpha -1} -1) } \bigg],\\
    &\stackrel{(b)}{\leq} 0,
\end{align*}
where $(a)$ follows from writing $\log(m) = \dfrac{\log(m^{\alpha -1})}{\alpha -1}$, and 
$(b)$ follows from the following three observations i) $\dfrac{x\log(x)}{x-1}$ is an increasing function, ii) $m^{\alpha -1} \geq (m-1)^{\alpha -1} $ and iii) the first term is positive. Note that $\log(m-1)$ is not defined for $m=1$ but for $m=1$ observe that $\Theta(\alpha)$ is always $0$.
Therefore $\forall \alpha \geq 2$ $h(m,g) \leq 0.5$, since for 
$\alpha = 2$, $\Theta(\alpha) = \dfrac{m^2 -m}{(m^2 -(m-1)^2 ) -m)} = 0.5$.
\end{proof}

\end{proof}